\documentclass[sigconf]{acmart}

\usepackage{balance}
\usepackage{amssymb,amsmath}
\usepackage{mathtools}
\usepackage{microtype}

\let\brace\undefined
\let\brack\undefined
\DeclarePairedDelimiter{\abs}{\lvert}{\rvert}

\DeclarePairedDelimiter{\floor}{\lfloor}{\rfloor}
\DeclarePairedDelimiter{\brack}{\lbrack}{\rbrack}
\DeclarePairedDelimiter{\brace}{\lbrace}{\rbrace}
\DeclarePairedDelimiter{\paren}{\lparen}{\rparen}

\newcommand{\N}{\mathbb{N}}
\renewcommand{\C}{\mathbb{C}}
\newcommand{\Z}{\mathbb{Z}}
\newcommand{\T}{\mathsf{T}}
\newcommand{\F}{\mathsf{F}}
\newcommand{\FF}{\mathbb{F}}
\newcommand{\gset}{\brace{\pm1,\pm i}}
\newcommand{\Aeven}{{A_{\mathrm{even}}}}
\newcommand{\Aodd}{{A_{\mathrm{odd}}}}
\newcommand{\Beven}{{B_{\mathrm{even}}}}
\newcommand{\Bodd}{{B_{\mathrm{odd}}}}
\newcommand{\Leven}{L_{\text{even}}}
\newcommand{\Lodd}{L_{\text{odd}}}
\newcommand{\LA}{L_{\text{A}}}
\let\Re\undefined
\let\Im\undefined
\DeclareMathOperator{\resum}{resum}
\DeclareMathOperator{\imsum}{imsum}
\DeclareMathOperator{\Re}{Re}
\DeclareMathOperator{\Im}{Im}
\DeclareMathOperator{\DFT}{DFT}
\newcommand{\maxorder}{25}

\begin{document}

\copyrightyear{2018}
\acmYear{2018}
\setcopyright{acmlicensed}
\acmConference[ISSAC '18]{2018 ACM International Symposium on Symbolic and Algebraic Computation}{July 16--19, 2018}{New York, NY, USA}
\acmBooktitle{ISSAC '18: 2018 ACM International Symposium on Symbolic and Algebraic Computation, July 16--19, 2018, New York, NY, USA}
\acmPrice{15.00}
\acmDOI{10.1145/3208976.3209006}
\acmISBN{978-1-4503-5550-6/18/07}


\title{Enumeration of Complex Golay Pairs via Programmatic SAT}

\author{Curtis Bright}
\affiliation{\institution{University of Waterloo}}

\author{Ilias Kotsireas}
\affiliation{\institution{Wilfrid Laurier University}}

\author{Albert Heinle}
\affiliation{\institution{University of Waterloo}}

\author{Vijay Ganesh}
\affiliation{\institution{University of Waterloo}}

\begin{abstract}
We provide a complete enumeration of all complex Golay pairs of length up to~$\maxorder$,
verifying that complex Golay pairs do not exist in lengths $23$ and~$25$
but do exist in length~$24$.
This independently verifies work done by F.~Fiedler in 2013~\cite{fiedler2013small}
that confirms the 2002 conjecture of Craigen, Holzmann, and Kharaghani~\cite{CHK:DM:2002}
that complex Golay pairs of length $23$ don't exist.
Our enumeration method relies on the recently proposed SAT+CAS paradigm of combining
computer algebra systems with SAT solvers to take advantage of the advances
made in the fields of symbolic computation and satisfiability checking.
The enumeration proceeds in two stages:
First, we use a fine-tuned computer program and functionality
from computer algebra systems to construct a list containing all sequences
which could appear as the first sequence in a complex Golay pair (up to equivalence).
Second, we use a programmatic SAT solver to
construct all sequences (if any) that pair off with the sequences constructed in the first stage
to form a complex Golay pair.
\end{abstract}


\maketitle

\section{Introduction}

The sequences which are now referred to \emph{Golay sequences} or \emph{Golay pairs} were first 
introduced by Marcel Golay in his groundbreaking 1949 paper~\cite{golay1949multi} on multislit spectrometry.
He later formally defined them in a~1961 paper~\cite{golay1961complementary}
where he referred to them as \emph{complementary series}.
Since then, Golay pairs and their generalizations have been widely studied for both their elegant
theoretical properties and a surprising number of practical applications.
For example, they have been applied to radar pulse compression~\cite{Hussain2014}, Wi-Fi networks~\cite{lomayev2017golay},
train wheel detection systems~\cite{1336500}, optical time domain reflectometry~\cite{nazarathy1989real},
and medical ultrasounds~\cite{nowicki2003application}.
Golay pairs consist of two sequences and the property that makes them
special is, roughly speaking, the fact that one sequence's ``correlation''
with itself is the inverse of the other sequence's ``correlation'' with itself;
see Definition~\ref{def:cgp} in Section~\ref{sec:background} for
the formal definition.

Although Golay defined his complementary series over an alphabet of $\brace{\pm1}$, later authors have generalized
the alphabet to include nonreal roots of unity such as 
the fourth root of unity $i=\sqrt{-1}$.
In this paper, we focus on the case where the alphabet
is $\brace{\pm1,\pm i}$. 
In this case the resulting sequence pairs
are sometimes referred to as \emph{4-phase} or \emph{quaternary} Golay pairs though we
will simply refer to them as \emph{complex} Golay pairs.
If a complex Golay pair of length $n$ exists then we say that $n$ is a \emph{complex Golay number}.

Complex Golay pairs have been extensively studied by many authors.
They were originally introduced in 1994 by Craigen
in order to expand the orders of Hadamard matrices attainable via (ordinary) Golay pairs~\cite{C:JCMCC:1994}.
In 1994, Holzmann and Kharaghani enumerated all complex Golay pairs up to length~$13$~\cite{HK:AJC:1994}.
In 2002, Craigen, Holzmann, and Kharaghani enumerated all complex Golay pairs to~$19$, reported
that~$21$ was not a complex Golay number, and conjectured that~$23$ was not a complex Golay number~\cite{CHK:DM:2002}.

In 2006, Fiedler, Jedwab, and Parker provided a construction which explained
the existence of all known complex Golay pairs whose lengths were a power
of~$2$~\cite{fiedler2008multi,fiedler2008framework},
including complex Golay pairs of length~$16$ discovered by Li and Chu~\cite{li2005more} to
not fit into a construction given by Davis and Jedwab~\cite{davis1999peak}.
In 2010, Gibson and Jedwab provided
a construction which explained the existence of all complex Golay pairs
up to length~$26$ and gave a table that listed the total number of complex Golay pairs up to length~$26$~\cite{gibson2011quaternary}.
This table was produced by the mathematician Frank Fiedler,
who described his enumeration method in a 2013 paper~\cite{fiedler2013small}
where he also reported that~$27$ and~$28$ are not complex Golay numbers.

In this paper we give an enumeration method which can be used to
verify the table produced by Fiedler that appears in Gibson and Jedwab's paper;
this table contains counts for the total number of complex Golay pairs and the total
number of sequences which appear as a member of a complex Golay pair.
We implemented our method 
and obtained counts up to length~$25$ after
about a day of computing on a cluster with~25 cores.
The counts we obtain match those in Fiedler's table in each case, increasing the
confidence that the enumeration was performed without error.
In addition, we also provide counts
for the total number of complex Golay pairs up to well-known equivalence operations%
~\cite{HK:AJC:1994} and explicitly make available
the sequences online~\cite{gcpweb}.  To our knowledge, this is the first time 
that explicit complex Golay pairs (and their counts up to equivalence) have been
published for lengths larger than $19$.
Lastly, we publicly release our code for enumerating complex Golay pairs
so that others may verify and reproduce our work; we were not able to find
any other code for enumerating complex Golay pairs which was publicly available.

Our result is of interest not only because of the verification we provide
but also because of the method we use to perform the verification.
The method proceeds in two stages.  In the first stage,
a fine-tuned computer program performs an exhaustive search
among all sequences which could possibly appear as the first
sequence in a complex Golay pair of a given length
(up to an equivalence defined in Section~\ref{sec:background}).
Several filtering theorems which we
describe in Section~\ref{sec:background} allow us to discard almost all
sequences from consideration.  To apply these filtering theorems we use
functionality from the computer algebra system \textsc{Maple}~\cite{Maple10} and the
mathematical library FFTW~\cite{frigo2005design}.  After this filtering is completed we have
a list of sequences of a manageable size such that 
the first sequence of every complex Golay pair of a given length
(up to equivalence) appears in the list.

In the second stage, we use the programmatic SAT solver~\textsc{MapleSAT}~\cite{liang2017empirical}
to determine which sequences from the first stage (if any) can be
paired up with another sequence to form a complex Golay pair.  A SAT
instance is constructed from each sequence found in the first stage
such that the SAT instance is satisfiable if and only if the
sequence is part of a complex Golay pair.  Furthermore, in the
case that the instance is satisfiable
a satisfying assignment determines a sequence which forms
the second half of a complex Golay pair.

This method combines both computer algebra 
and SAT solving and 
is of interest in its own right because it links the two previously
separated fields of symbolic computation and satisfiability checking.
Recently there has been interest in combining methods from both fields
to solve computational problems as demonstrated by the
SC$^2$ project~\cite{sc2,abraham2015building}.  Our work fits into
this paradigm and to our knowledge is the first application of a SAT solver
to search for complex Golay pairs,
though previous work exists which uses
a SAT solver to search for other types of
complementary sequences~\cite{bright2016mathcheck,brightthesis,DBLP:journals/jar/ZulkoskiBHKCG17,bright2017sat+}.

\section{Background on Complex Golay Pairs}\label{sec:background}

In this section we present the background necessary to describe our
method for enumerating complex Golay pairs.
First, we require some preliminary definitions
to define what complex Golay pairs are.
Let $\overline{x}$ denote the complex conjugate of $x$
(this is just the multiplicative inverse of $x$
when $x$ is $\pm1$ or $\pm i$).

\begin{definition}[cf.~\cite{kotsireas2013algorithms}]
The \emph{nonperiodic autocorrelation function}
of a sequence $A = [a_0, \dotsc, a_{n-1}] \in \C^n$ of length $n\in \N$ is 
\[ N_A(s) \coloneqq \sum_{k=0}^{n-s-1} a_k\overline{a_{k+s}}, \qquad s=0,\dotsc, n-1 . \]
\end{definition}

\begin{definition}\label{def:cgp}
  A pair of sequences $(A,B)$ with $A$ and $B$ in $\gset^n$
  are called a \emph{complex Golay pair}
  if the sum of their nonperiodic autocorrelations is a constant zero for $s\neq0$, i.e.,
  \begin{equation*}
  N_A(s) + N_B(s) = 0 \qquad\text{for}\qquad s = 1, \dotsc, n-1 . 
  \end{equation*}
\end{definition}

Note that if $A$ and $B$ are in $\gset^n$
then $N_A(0)+N_B(0)=2n$ by the definition of
the complex nonperiodic autocorrelation function
and the fact that
$x\overline{x}
=1$ if $x$ is $\pm1$ or $\pm i$, explaining why $s\neq0$ in Definition~\ref{def:cgp}.
\begin{example}
$([1,1,-1], [1, i, 1])$ is a complex Golay pair.
\end{example}

\subsection{Equivalence operations}\label{subsec:equiv}

There are certain invertible operations which preserve the property of
being a complex Golay pair when applied to a sequence pair $(A,B)$.
These are summarized in the following proposition.

\begin{proposition}[cf.~\cite{CHK:DM:2002}]
\label{prop:equivGolay}
Let\/ $([a_0,\ldots,a_{n-1}]$, $[b_0,\ldots,b_{n-1}])$ be a complex
Golay pair. The following are then also complex Golay pairs:
\begin{enumerate}
\item[E1.] (Reversal)\/ $([a_{n-1},\dotsc,a_0], [b_{n-1},\dotsc,b_0])$.
\item[E2.] (Conjugate Reverse $A$)\/ $([\overline{a_{n-1}},\ldots,\overline{a_0}], [b_0,\ldots,b_{n-1}])$.
\item[E3.] (Swap)\/ $([b_0,\dotsc,b_{n-1}], [a_0,\dotsc,a_{n-1}])$. 
\item[E4.] (Scale $A$)\/ $([ia_0,\dotsc,ia_{n-1}], [b_0,\dotsc,b_{n-1}])$. 
\item[E5.] (Positional Scaling)\/ $(i\star A, i\star B)$
where\/ $c\star(x_0,\dotsc,x_{n-1})\coloneqq(x_0,cx_1,c^2x_2,\dotsc,c^{n-1}x_{n-1})$.
\end{enumerate}
\end{proposition}

\begin{definition}
We call two complex Golay pairs $(A, B)$ and $(A', B')$ \emph{equivalent}
if $(A', B')$ can be obtained from
$(A, B)$ using the transformations described in Proposition~\ref{prop:equivGolay}.
\end{definition}

\subsection{Useful properties and lemmas}

In this subsection we prove some useful properties that complex
Golay pairs satisfy and which will be exploited by our method
for enumerating complex Golay pairs.
The first lemma provides a fundamental
relationship that all complex Golay pairs must
satisfy.  To conveniently state it we use the following definition.

\begin{definition}[cf.~\cite{CHK:DM:2002}]
The \emph{Hall polynomial} of the sequence $A\coloneqq[a_0,\dotsc,a_{n-1}]$
is defined to be $h_A(z)\coloneqq a_0 + a_1 z + \dotsb + a_{n-1}z^{n-1}\in\C[z]$.
\end{definition}

\begin{lemma}[cf.~\cite{paterson2000generalized}]\label{lem:hall}
  Let\/ $(A, B)$ be a complex Golay pair.
  For every\/ $z \in \C$ with\/ $\abs{z} = 1$,
  we have
  \begin{equation*}\abs{h_A(z)}^2+\abs{h_B(z)}^2 = 2n.\end{equation*}
\end{lemma}
\begin{proof}
Since $\abs{z}=1$ we can write $z=e^{i\theta}$ for some $0\leq\theta<2\pi$.
Similar to the fact pointed out in~\cite{kharaghani2005hadamard},
using Euler's identity one can derive the following expansion:
\[ \abs{h_A(z)}^2 = N_A(0) + 2\sum_{j=1}^{n-1}\paren[\big]{\Re(N_A(j))\cos(\theta j) + \Im(N_A(j))\sin(\theta j)} . \]
Since $A$ and $B$ form a complex Golay pair, by definition one has that
$\Re(N_A(j)+N_B(j))=0$ and $\Im(N_A(j)+N_B(j))=0$ and then
\[ \abs{h_A(z)}^2 + \abs{h_B(z)}^2 = N_A(0) + N_B(0) = 
2n . \qedhere \]
\end{proof}

This lemma is highly useful as a condition for filtering sequences which
could not possibly be part of a complex Golay pair, as explained
in the following corollary.

\begin{corollary}\label{cor:filter}
Let\/ $A\in\C^n$, $z\in\C$ with\/ $\abs{z}=1$, and\/ $\abs{h_A(z)}^2>2n$.
Then\/ $A$ is not a member of a complex Golay pair.
\end{corollary}
\begin{proof}
Suppose the sequence $A$ was a member of a complex Golay pair whose other
member was the sequence $B$.
Since $\abs{h_B(z)}^2\geq0$, we must have $\abs{h_A(z)}^2+\abs{h_B(z)}^2>2n$,
in contradiction to Lemma~\ref{lem:hall}.
\end{proof}

In~\cite{fiedler2013small}, Fiedler derives the following extension of Lemma~\ref{lem:hall}.
Let $\Aeven$ be identical to $A$ with the entries of odd index replaced by zeros and let $\Aodd$ be
identical to $A$ with the entries of even index replaced by zeros.

\begin{lemma}[cf.~\cite{fiedler2013small}]\label{lem:fiedler}
  Let\/ $(A, B)$ be a complex Golay pair.
  For every\/ $z \in \C$ with\/ $\abs{z} = 1$,
  we have
  \begin{equation*}
  \abs{h_\Aeven(z)}^2+\abs{h_\Aodd(z)}^2+\abs{h_\Beven(z)}^2+\abs{h_\Bodd(z)}^2 = 2n . \qedhere
  \end{equation*}
\end{lemma}

\begin{proof}
The proof proceeds as in the proof of Lemma~\ref{lem:hall}, except
that one instead obtains that $\abs{h_{\Aeven}(z)}^2 + \abs{h_{\Aodd}(z)}^2$ is equal to
\[ N_A(0) + 2\sum_{\substack{j=1\\\text{$j$ even}}}^{n-1}\paren[\big]{\Re(N_A(j))\cos(\theta j) + \Im(N_A(j))\sin(\theta j)} . \qedhere \]
\end{proof}

\begin{corollary}\label{cor:fiedler}
Let\/ $A\in\C^n$, $z\in\C$ with\/ $\abs{z}=1$, and\/ $\abs{h_{A'}(z)}^2>2n$
where\/ $A'$ is either\/ $\Aeven$ or\/ $\Aodd$.
Then\/ $A$ is not a member of a complex Golay pair.
\end{corollary}

\begin{proof}
If either $\abs{h_\Aeven(z)}^2>2n$
or $\abs{h_\Aodd(z)}^2>2n$ then the identity in Lemma~\ref{lem:fiedler} cannot hold.
\end{proof}

The next lemma is useful because it allows us to write $2n$
as the sum of four integer squares.  It is stated
in~\cite{HK:AJC:1994} using a different notation; we use the notation
$\resum(A)$ and $\imsum(A)$ to represent the real and imaginary parts of
the sum of the entries of $A$.
For example, if $A\coloneqq[1,i,-i,i]$ then $\resum(A)=\imsum(A)=1$. 

\begin{lemma}[cf.~\cite{HK:AJC:1994}]\label{lem:decomp}
  Let\/ $(A, B)$ be a complex Golay sequence
  pair. Then
  \[ {\resum(A)}^2 + {\imsum(A)}^2 + {\resum(B)}^2 + {\imsum(B)}^2 = 2n . \]
\end{lemma}
\begin{proof}
Using Lemma~\ref{lem:hall} with $z=1$ we have
\[ \abs{\resum(A)+\imsum(A)i}^2 + \abs{\resum(B)+\imsum(B)i}^2 = 2n . \]
Since $\abs{\resum(X)+\imsum(X)i}^2 = \resum(X)^2 + \imsum(X)^2$ the result follows.
\end{proof}

The next lemma provides some normalization conditions which can be used when searching
for complex Golay pairs up to equivalence.
Since all complex Golay pairs $(A',B')$ which are equivalent to
a complex Golay pair $(A,B)$ can easily
be generated from $(A,B)$, it suffices to search for complex Golay pairs up to equivalence.

\begin{lemma}[cf.~\cite{fiedler2013small}]\label{lem:normalize}
  Let\/ $(A', B')$ be a complex Golay pair.
  Then\/ $(A', B')$ is equivalent to a complex Golay pair\/
  $(A, B)$ with\/ $a_0=a_1=b_0=1$ and\/ $a_2\in\brace{\pm1, i}$.
\end{lemma}
\begin{proof}
We will transform a given complex Golay sequence pair $(A', B')$ into 
an equivalent normalized one
using the equivalence operations of
Proposition~\ref{prop:equivGolay}.
To start with, let $A\coloneqq A'$ and $B\coloneqq B'$.


First, we ensure that $a_0=1$.  To do this, we apply operation E4 (scale $A$) enough times
until $a_0=1$.

Second, we ensure that $a_1=1$.  To do this, we apply operation E5 (positional scaling) enough times
until $a_1=1$; note that E5 does not change $a_0$.

Third, we ensure that $a_2\neq -i$.  If it is, we apply operation E1 (reversal) and E2 (conjugate reverse $A$)
which has the effect of keeping $a_0=a_1=1$ and setting $a_2=i$.

Last, we ensure that $b_0=1$.  To do this, we apply operation E3 (swap) and then operation E4 (scale $A$) enough times
so that $a_0=1$ and then operation E3 (swap) again.  This has the effect of not changing $A$ but setting $b_0=1$.
\end{proof}

\subsection{Sum-of-squares decomposition types}\label{subseq:decomp}

A consequence of Lemma~\ref{lem:decomp} is
that every complex Golay pair generates a decomposition of $2n$ into a sum of four integer squares.
In fact, it typically generates several decompositions of $2n$ into a sum of four squares.
Recall that $i\star A$ denotes positional scaling by~$i$ (operation E5) on the sequence~$A$.
If $(A,B)$ is a complex Golay pair then applying operation E5 to this pair $k$ times 
shows that $(i^k\star A,i^k\star B)$ 
is also a complex Golay pair.  
By using Lemma~\ref{lem:decomp} on these complex Golay pairs
one obtains the fact that $2n$ can be decomposed as the sum of 
four integer squares as
\begin{equation*}
\resum\paren{i^k\star A}^2 + \imsum\paren{i^k\star A}^2 + \resum\paren{i^k\star B}^2 + \imsum\paren{i^k\star B}^2 .
\end{equation*}
For $k>3$ this produces no new decompositions but in general
for $k=0$, $1$, $2$, and~$3$ this produces four distinct
decompositions of $2n$ into a sum of four squares.

With the help of a computer algebra system (CAS) one can
enumerate every possible way that $2n$ may be written as a sum of four integer squares.  For example,
when $n=23$ one has $0^2+1^2+3^2+6^2=2\cdot23$ and $1^2+2^2+4^2+5^2=2\cdot23$ as well as all
permutations of the squares and negations of the integers being squared.
During the first stage of our enumeration method only the first sequence of a complex Golay pair
is known, so at that stage we cannot compute its whole
sums-of-squares decomposition.
However, it is still possible to filter some sequences from consideration based on analyzing
the two known terms in the sums-of-squares decomposition.

For example, say that $A$ is the
first sequence in a potential complex Golay pair of length $23$ with $\resum(A)=0$ and $\imsum(A)=5$.
We can immediately discard $A$ from consideration because 
there is no way to chose the $\resum$ and $\imsum$ of $B$ to complete the sums-of-squares
decomposition of $2n$, i.e., there are no integer solutions $(x,y)$ of $0^2+5^2+x^2+y^2=2n$.

\section{Enumeration Method}\label{sec:method}

In this section we describe in detail the method we used to
perform a complete enumeration of all complex Golay pairs up to length~$\maxorder$.
Given a length $n$ our goal is to find all $\gset$ sequences $A$ and $B$
of length $n$ such that $(A,B)$ is a complex Golay pair.

\subsection{Preprocessing: Enumerate possibilities for $\boldsymbol\Aeven$ and $\boldsymbol\Aodd$}

The first step of our method uses Fiedler's trick of considering
the entries of $A$ of even index separately from the entries of $A$ of odd index.
There are approximately $n/2$ nonzero entries in each of $\Aeven$ and $\Aodd$
and there are four possible values for each nonzero entry.  Therefore there
are approximately $2\cdot4^{n/2}=2^{n+1}$ possible sequences to check in this
step.  Additionally, by Lemma~\ref{lem:normalize} we may assume the first nonzero
entry of both $\Aeven$ and $\Aodd$ is $1$ and that the second nonzero
entry of $\Aeven$ is not $-i$, decreasing the number of sequences
to check in this step by more than a factor of $4$.
It is quite feasible to perform a brute-force search
through all such sequences when $n\approx30$.

We apply Corollary~\ref{cor:fiedler} to every possibility for $\Aeven$ and $\Aodd$.
There are an infinite number of possible $z\in\C$ with $\abs{z}=1$, so
we do not attempt to apply Corollary~\ref{cor:fiedler} using all such $z$.
Instead we try a sufficiently large number of $z$ so that 
in the majority of cases for which a $z$ exists with $\abs{h_{A'}(z)}^2>2n$
(where $A'$ is either $\Aeven$ or $\Aodd$) we in fact find such a $z$.
In our implementation we chose to take $z$ to be $e^{2\pi ij/N}$
where $N\coloneqq2^{14}$ and $j=0$, $\dotsc$, $N-1$.

At the conclusion of this step we will have two lists: one list $\Leven$
of the $\Aeven$ which were not discarded and one list $\Lodd$ of the $\Aodd$
which were not discarded.

\subsection{Stage 1: Enumerate possibilities for $\boldsymbol A$}

We now enumerate all possibilities for $A$ by joining the
possibilities for $\Aeven$ with the possibilities for $\Aodd$.
For each $A_1\in\Lodd$ and $A_2\in\Leven$ we form the sequence $A$
by letting the $k$th entry of $A$ be either the $k$th entry of $A_1$
or $A_2$ (whichever is nonzero).  Thus the entries of $A$ are either
$\pm1$ or $\pm i$ and therefore $A$ is a valid candidate
for the first sequence of a complex Golay pair of length $n$.

At this stage we now use the filtering result of Corollary~\ref{cor:filter} and 
the sums-of-squares decomposition result of Lemma~\ref{lem:decomp} to perform
more extensive filtering on the sequences $A$ which we formed above.
In detail, our next filtering check proceeds as follows:
Let $R_k\coloneqq\resum(i^k\star A)$ and $I_k\coloneqq\imsum(i^k\star A)$.
By using a Diophantine equation solver we check if the Diophantine equations
\[  R_k^2 + I_k^2 + x^2 + y^2 = 2n  \]
are solvable in integers $(x,y)$ for $k=0$, $1$, $2$, $3$.
As explained in Section~\ref{subseq:decomp}, if any of these equations
have no solutions then $A$ cannot be a member of a complex Golay pair
and can be ignored.
Secondly, we use Corollary~\ref{cor:filter} with $z$ chosen to
be $e^{2\pi ij/N}$ for $j=0$, $\dotsc$, $N-1$ where $N\coloneqq 2^7$
(we use a smaller value of $N$ than in the preprocessing step
because in this case there are a larger number of sequences
which we need to apply the filtering condition on).

If $A$ passes both filtering conditions then we add it to a list $\LA$ and
try the next value of $A$ until no more possibilities remain.
At the conclusion of this stage we will have a list of sequences $\LA$
which could potentially be a member of a complex Golay pair.
By construction, the first member of all complex Golay pairs
(up to the equivalence described in Lemma~\ref{lem:normalize})
of length $n$ will be in $\LA$.

\subsection{Stage 2: Construct the second sequence $\boldsymbol B$ from $\boldsymbol A$}\label{sec:stage2}

In the second stage we take as input the list $\LA$ generated
in the first stage, i.e., a list of the sequences $A$ that were not filtered by
any of the filtering theorems we applied.  For each $A\in\LA$ we attempt to
construct a second sequence $B$ such that $(A,B)$ is a complex Golay pair.
We do this by generating a SAT instance which encodes the property of
$(A,B)$ being a complex Golay pair where
the entries of $A$ are known and the entries
of $B$ are unknown and encoded using Boolean variables.
Because there are four possible values for each entry of $B$
we use two Boolean variables to encode each entry. 
Although the exact encoding used is arbitrary, we fixed the following
encoding in our implementation, where the variables $v_{2k}$ and $v_{2k+1}$
represent $b_k$, the $k$th entry of $B$:
\[ \begin{array}{@{\;\;}c@{\;\;}c@{\;\;}|@{\;\;}c@{\;\;}}
v_{2k} & v_{2k+1} & b_k \\ \hline
\F & \F & 1 \\
\F & \T & -1 \\
\T & \F & i \\
\T & \T & -i
\end{array} \]

To encode the property that $(A,B)$ is a complex Golay pair in out SAT instance
we add the conditions which define $(A,B)$ to be a complex Golay pair, i.e.,
\begin{equation*}
N_A(s) + N_B(s) = 0 \qquad\text{for}\qquad s = 1, \dotsc, n-1 . 
\end{equation*}
These equations could be encoded using clauses in conjunctive normal form
(for example by constructing
logical circuits to perform complex multiplication and addition and then
converting those circuits into CNF clauses).  However, we found that a much
more efficient and convenient method was to use a \emph{programmatic} SAT solver.

The concept of a programmatic SAT solver was first introduced
in~\cite{ganesh2012lynx} where a programmatic SAT solver was
shown to be more efficient than a standard SAT solver when solving
instances derived from RNA folding problems.
More recently, a programmatic SAT solver was also shown to be
useful when searching for Williamson matrices~\cite{bright2017sat+}.
Generally, programmatic SAT solvers perform well when there is
domain-specific knowledge about the problem being
solved that cannot easily be encoded into SAT instances directly
but can be used to learn facts about potential solutions
which can help guide the solver in its search.

Concretely, a programmatic SAT solver is compiled with a piece
of code which encodes a property that a solution of the SAT instance
must satisfy.  Periodically the SAT solver will run this code
while performing its search and if the current partial
assignment violates a property that is expressed in the provided
code then a conflict clause is generated encoding this fact.
The conflict clause is added to the SAT solver's database
of learned clauses where it is used to increase the efficiency
of the remainder of the search.
The reason these clauses can be so useful is because they can
encode facts which the SAT solver would have no way of learning
otherwise, since the SAT solver has no knowledge of the domain of
the problem.

Not only does this paradigm allow the SAT solver to perform its
search more efficiently, it also allows instances
to be much more expressive.  Under this framework SAT instances do not
have to consist solely of Boolean formulas in conjunctive normal
form (the typical format of SAT instances) but can consist of
clauses in conjunctive normal form combined with a piece of
code that \emph{programmatically} expresses clauses.
This extra expressiveness is also a feature of SMT solvers,
though SMT solvers typically require more overhead to use.  Additionally,
one can compile \emph{instance-specific} programmatic SAT solvers
which are tailored to perform searches for a specific class of
problems.

For our purposes we use a programmatic
SAT solver tailored to search for
sequences $B$ that when paired with a given sequence~$A$
form a complex Golay pair.
Each instance will contain the $2n$ variables $v_0$, $\dotsc$, $v_{2n-1}$
that encode the entries of $B$ as previously specified.
In detail, the code given to the SAT solver does the following:
\begin{enumerate}
\item Compute and store the values $N_A(k)$ for $k=1$, $\dotsc$, $n-1$.
\item Initialize $s$ to $n-1$.  This will be a variable which
controls which autocorrelation condition we are currently examining.
\item Examine the current partial assignment to $v_0$, $v_1$,
$v_{2n-2}$, and $v_{2n-1}$.  If all these values have been assigned
then we can determine the values of $b_0$ and $b_{n-1}$.
From these values we compute
$N_B(s) = b_0\overline{b_{n-1}}$.
If $N_A(s)+N_B(s)\neq0$ then $(A,B)$ cannot be a complex Golay pair
(regardless of the values of $b_1$, $\dotsc$, $b_{n-2}$) and therefore
we learn a conflict clause which says that $b_0$ and $b_{n-1}$ cannot
both be assigned to their current values.  More explicitly,
if $v_k^{\text{cur}}$ represents the literal $v_k$ when $v_k$
is currently assigned to true and the literal $\lnot v_k$ when
$v_k$ is currently assigned to false
we learn the clause
\[ \lnot(v_0^{\text{cur}}\land v_1^{\text{cur}}\land v_{2n-2}^{\text{cur}}\land v_{2n-1}^{\text{cur}}) . \]
\item Decrement $s$ by $1$ and repeat the previous step,
computing $N_B(s)$ if the all the $b_k$ which appear
in its definition have known values.
If $N_A(s)+N_B(s)\neq0$ then learn a clause preventing the values
of $b_k$ which appear in the definition of $N_B(s)$ from being assigned
the way that they currently are.  Continue to repeat this step until $s=0$.
\item If all values of $B$ are assigned but no clauses
have been learned then output the complex Golay pair $(A,B)$.
If an exhaustive search is desired, learn a clause which prevents
the values of $B$ from being assigned the way they currently are;
otherwise learn nothing and return control to the SAT solver.
\end{enumerate}
For each $A$ in the list $\LA$ from stage~1 we run a SAT solver with the above
programmatic code; the list of all outputs $(A,B)$ in step~(5) shown
above now form a complete list of complex Golay pairs of length $n$
up to the equivalence given in Lemma~\ref{lem:normalize}.  In fact,
since Lemma~\ref{lem:normalize} says that we can set $b_0=1$ we can
assume that both $v_0$ and $v_1$ are always set to false.
In other words, we can add the two clauses $\lnot v_0$
and $\lnot v_1$ into our SAT instance without omitting any
complex Golay pairs up to equivalence.

\subsection{Postprocessing: Enumerating all complex Golay pairs}

At the conclusion of the second stage we have obtained a list of complex
Golay pairs of length $n$ such that every complex Golay pair of length $n$
is equivalent to some pair in our list.  However, because we have not
accounted for all the equivalences in Section~\ref{subsec:equiv} some pairs
in our list may be equivalent to each other.  In some sense such pairs
should not actually be considered distinct, so to count how many distinct
complex Golay pairs exist in length $n$ we would like
to find and remove pairs which are equivalent from the list.
Additionally, to verify the counts given in~\cite{gibson2011quaternary}
it is necessary to produce a list which contains
\emph{all} complex Golay pairs.  We now describe an algorithm which
does both, i.e., it produces a list of all complex Golay pairs as 
well as a list of all inequivalent complex Golay pairs.

In detail, our algorithm performs the following steps:
\begin{enumerate}
\item Initialize $\Omega_\text{all}$ to be the set of complex Golay pairs
generated in stage 2.  This variable will be a set that will be
populated with and eventually contain all complex Golay pairs of length~$n$.
\item Initialize $\Omega_{\text{inequiv}}$ to be the empty set.  This variable
will be a set that will be populated with and eventually contain all inequivalent
complex Golay pairs of length~$n$.
\item For each $(A,B)$ in $\Omega_{\text{all}}$:
\begin{enumerate}
\item If $(A,B)$ is already in $\Omega_{\text{inequiv}}$ then skip this $(A,B)$
and proceed to the next pair $(A,B)$ in $\Omega_{\text{all}}$.
\item Initialize $\Gamma$ to be the set containing $(A,B)$.
This variable will be a set that will be populated with and eventually
contain all complex Golay pairs equivalent to $(A,B)$.
\item For every $\gamma$ in $\Gamma$ add
$\operatorname{E1}(\gamma)$, $\dotsc$, $\operatorname{E5}(\gamma)$ to $\Gamma$.
Continue to do this until every pair in $\Gamma$ has been examined and no new
pairs are added to $\Gamma$.
\item Add $(A,B)$ to $\Omega_{\text{inequiv}}$ and add all pairs in $\Gamma$ to
$\Omega_{\text{all}}$.
\end{enumerate}
\end{enumerate}

After running this algorithm listing the members of $\Omega_{\text{all}}$
gives a list of all complex Golay pairs of length~$n$ and
listing the members of $\Omega_{\text{inequiv}}$
gives a list of all inequivalent complex Golay pairs of length~$n$.
At this point we can also construct the complete list of sequences
which appear in any complex Golay pair of length~$n$.
To do this it suffices to add $A$ and $B$ to a new
set $\Omega_{\text{seqs}}$ for each $(A,B)\in\Omega_{\text{all}}$.

\subsection{Optimizations}\label{sec:optimization}

Although the method described will correctly enumerate all complex Golay
pairs of a given length $n$, for the benefit of potential implementors
we mention a few optimizations which we found helpful.

In stage 1 we check if Diophantine equations of the form
\[ R^2 + I^2 + x^2 + y^2 = 2n \tag{$*$}\label{eq:dioeq} \]
are solvable in integers $(x,y)$ where $R$ and $I$ are given.
CAS functions like \texttt{PowersRepresentations} in \textsc{Mathematica}
or \texttt{nsoks} in \textsc{Maple}~\cite{nsoks}
can determine all ways of writing $2n$ as a sum of four integer squares.
From this information we construct a Boolean two dimensional array
$D$ such that $D_{\abs{R},\abs{I}}$ is true if and only if~\eqref{eq:dioeq}
has a solution, making the check for solvability a fast lookup.
In fact, one need only construct the lookup table
for $R$ and $I$ with $R+I\equiv n\pmod{2}$ as the following lemma shows.
\begin{lemma}
Suppose\/ $R$ and\/ $I$ are the $\resum$ and $\imsum$ of a sequence\/
$X\in\brace{\pm1,\pm i}^n$.  Then\/ $R+I\equiv n\pmod{2}$.
\end{lemma}
\begin{proof}
Let $\#_c$ denote the number of entries in $X$ with value $c$.  Then
\[ R + I = (\#_1 - \#_{-1}) + (\#_{i} - \#_{-i}) \equiv \#_1 + \#_{-1} + \#_{i} + \#_{-i} \pmod{2} \]
since $-1\equiv1\pmod{2}$.
The quantity on the right is $n$ since there are $n$ entries in $X$.
\end{proof}

In stage 1 we check if $\abs{h_A(z)}^2>2n$ where $z=e^{2\pi ij/N}$ for
$j=0$, $\dotsc$, $N-1$ with $N=2^7$.  However, we found that it was more efficient
to not check the condition for each $j$ in ascending order (i.e., for each~$z$
in ascending complex argument)
but to first perform the check on points $z$ with larger spacing between them.
In our implementation we first assigned~$N$
to be~$2^3$ and performed the check for odd $j=1$, $3$, $\dotsc$, $N-1$.
Following this we doubled~$N$ and again performed the check for odd~$j$, proceeding in this
matter until all points~$z$ had been checked.  (This ignores checking the
condition when $z=i^k$ for some $k$ but that is desirable since
in those cases $\abs{h_A(i^k)}^2=\resum(i^k\star A)^2+\imsum(i^k\star A)^2$
and the sums-of-squares condition is a strictly stronger filtering method.)

In the preprocessing step and stage~1 it is necessary to evaluate
the Hall polynomial $h_{A'}$ or $h_A$ at roots of unity $z=e^{2\pi ij/N}$
and determine its squared absolute value.
The fastest way we found of doing this used the discrete Fourier transform.
For example, let $A'$ be the sequence $\Aeven$, $\Aodd$, or $A$
under consideration but padded with trailing zeros
so that $A'$ is of length~$N$.
By definition of the discrete Fourier transform we have that
\[ \DFT(A') = \brack*{h_{A'}\paren[\big]{e^{2\pi ij/N}}}_{j=0}^{N-1} . \]
Thus, we determine the values of $\abs{h_{A'}(z)}^2$ by
taking the squared absolute values of the entries of $\DFT(A')$.  If
$\abs{h_{A'}(z)}^2>2n$ for some~$z$ then by Corollary~\ref{cor:filter}
or Corollary~\ref{cor:fiedler}
we can discard $A'$ from consideration.  To guard against potential
inaccuracies introduced by the algorithms used to compute the DFT we
actually ensure that $\abs{h_{A'}(z)}^2>2n+\epsilon$ for some tolerance~$\epsilon$
which is small but larger than the accuracy that the DFT is computed to
(e.g., $\epsilon=10^{-3}$).

In the preprocessing step before setting $N\coloneqq2^{14}$ we first set $N\coloneqq n$
and perform the rest of the step as given.  The advantage of first
performing the check with a smaller value of $N$ is that the discrete
Fourier transform of $A'$ can be computed faster.  Although
the check with $N=n$ is a less effective filter, it often succeeds and whenever it
does it allows us to save time by not performing the more costly longer DFT.

In stage 1 our application of Corollary~\ref{cor:filter} requires
computing $\abs{h_A(z)}^2$ where $z=e^{2\pi ij/N}$ for $j=0$, $\dotsc$, $N-1$.
Noting that
\[ h_A(z) = h_{\Aeven}(z) + h_{\Aodd}(z) \]
one need only compute $h_{\Aeven}(z)$ and $h_{\Aodd}(z)$
for each 
each $\Aeven$ and $\Aodd$ generated in the preprocessing step
and once those are known $h_A(z)$ can be found by a simple addition.

In stage 2 one can also include properties that complex Golay sequences
must satisfy in the code compiled with the programmatic SAT solver.
As an example of this, we state the following proposition which was new to the
authors and does not appear to have been previously published.
\begin{proposition}\label{prop:prod}
Let\/ $(A,B)$ be a complex Golay pair.  Then
\[ a_k a_{n-k-1} b_k b_{n-k-1} = \pm 1 \qquad\text{for}\qquad\text{$k=0$, $\dotsc$, $n-1$}. \]
\end{proposition}

To prove this, we use the following simple lemma.
\begin{lemma}\label{lem:sumsofi}
Let\/ $c_k\in\Z_4$ for\/ $k=0$, $\dotsc$, $n-1$.  Then
\[ \sum_{k=0}^{n-1}i^{c_k}=0 \qquad\text{implies}\qquad \sum_{k=0}^{n-1} c_k\equiv0\pmod{2}. \]
\end{lemma}
\begin{proof}
Let $\#_c$ denote the number of $c_k$ with value $c$.
Note that the sum on the left implies that $\#_0=\#_2$ and $\#_1=\#_3$
because the $1$s must cancel with the $-1$s and the $i$s must cancel with the $-i$s.
Then
$\sum_{k=0}^{n-1}c_k=\#_1+2\#_2+3\#_3\equiv\#_1+\#_3\equiv2\#_1\equiv0\pmod{2}$.
\end{proof}
We now prove Proposition~\ref{prop:prod}.
\begin{proof}
Let $c_k$, $d_k\in\Z_4$ be such
that $a_k=i^{c_k}$ and $b_k=i^{d_k}$.  Using this notation
the multiplicative equation from Proposition~\ref{prop:prod} becomes
the additive congruence
\[ c_k + c_{n-k-1} + d_k + d_{n-k-1} \equiv 0 \pmod{2} . \label{eq:star}\tag{$*$} \]
Since $(A,B)$ is a complex Golay pair, the autocorrelation equations give us
\[ \sum_{k=0}^{n-s-1}\paren[\Big]{i^{c_k-c_{k+s}}+i^{d_k-d_{k+s}}} = 0 \]
for $s=1$, $\dotsc$, $n-1$.  Using Lemma~\ref{lem:sumsofi} and the fact that
$-1\equiv1\pmod{2}$ gives
\[ \sum_{k=0}^{n-s-1}\paren[\big]{c_k + c_{k+s} + d_k + d_{k+s}} \equiv 0 \pmod{2} \]
for $s=1$, $\dotsc$, $n-1$.  With $s=n-1$ one immediately derives~\eqref{eq:star}
for $k=0$.  With $s=n-2$ and~\eqref{eq:star} for $k=0$ one derives~\eqref{eq:star}
for $k=1$.  Working inductively in this manner one derives~\eqref{eq:star} for all $k$.
\end{proof}

In short, Proposition~\ref{prop:prod} tells us that an even number of
$a_k$, $a_{n-k-1}$, $b_k$, and $b_{n-k-1}$ are real for each $k=0$, $\dotsc$, $n-1$.
For example, if exactly one of $a_k$ and $a_{n-k-1}$ is real then exactly
one of $b_k$ and $b_{n-k-1}$ must also be real.
In this case, using our encoding from
Section~\ref{sec:stage2} we can add the clauses
\[ (v_{2k}\lor v_{2(n-k-1)})\land(\lnot v_{2k}\lor \lnot v_{2(n-k-1)}) \]
to our SAT instance.
These clauses say that exactly one of $v_{2k}$ and $v_{2(n-k-1)}$ is true.

\section{Results}

In order to provide a verification of the counts from~\cite{gibson2011quaternary}
we implemented the enumeration method described in Section~\ref{sec:method}.
The preprocessing step was performed by a C program
and used the mathematical library FFTW~\cite{frigo2005design} for computing
the values of $h_{A'}(z)$ as described in Section~\ref{sec:optimization}. 
Stage~1 was performed by a C++ program,
used FFTW for computing the values of $h_A(z)$
and a \textsc{Maple} script~\cite{nsoks}
for determining the solvability of the Diophantine equations 
given in Section~\ref{sec:stage2}.  Stage~2 was performed
by the programmatic SAT solver \textsc{MapleSAT}~\cite{liang2017empirical}.
The postprocessing step was performed by a Python script.

\begin{table}\begin{tabular}{cccc}
& \multicolumn{3}{c}{Total CPU Time in hours} \\
$n$ & Preproc. & Stage~1 & Stage~2 \\ \hline
17         & 0.00       & 0.01       & 0.06       \\ 
18         & 0.01       & 0.03       & 0.23       \\ 
19         & 0.01       & 0.07       & 0.18       \\ 
20         & 0.02       & 0.35       & 0.43       \\ 
21         & 0.04       & 1.93       & 1.89       \\ 
22         & 0.08       & 9.58       & 1.11       \\ 
23         & 0.15       & 42.01      & 3.02       \\ 
24         & 0.32       & 81.42      & 5.23       \\ 
25         & 0.57       & 681.31     & 20.51      
\end{tabular}
\caption{The time used to run the various stages of our algorithm in lengths $\boldsymbol{17\leq n\leq25}$.}
\label{tbl:timings}\end{table}

We ran our implementation on a cluster of machines running CentOS 7 and
using Intel Xeon E5-2683V4 processors running at $2.1$~GHz and
using at most 300MB of RAM.
To parallelize the work in each length~$n$ we split $\Lodd$ into 25 pieces
and used 25 cores to complete stages~1 and~2 of the algorithm.
Everything in the stages proceeded exactly as before
except that in stage~1 the list $\Lodd$
was 25 times shorter than it would otherwise be, which
allowed us to complete the first stages 20.7 times faster
and the second stages 23.9 times faster. 
The timings for the preprocessing step and the two stages of our algorithm are
given in Table~\ref{tbl:timings}; the timings for the postprocessing step
were negligible.  The times are given as the total amount of CPU time used
across all 25 cores.
Our code is available online as a part of the \textsc{MathCheck} project
and we have also made available
the resulting enumeration of complex Golay pairs~\cite{gcpweb}.

The sizes of the lists $\Leven$ and $\Lodd$ computed in the preprocessing step
and the size of the list $\LA$ computed in stage~1
are given in Table~\ref{tbl:listsizes} for all lengths in which we completed
a search.  Without applying any filtering $\LA$ would have size $4^n$ so
Table~\ref{tbl:listsizes} demonstrates the power of the criteria we used to
perform filtering; typically far over $99.99\%$ of possible sequences $A$ are
filtered from~$\LA$.
The generated SAT instances had $2n$ variables
(encoding the entries $b_0$, $\dotsc$, $b_{n-1}$),
$2$ unit clauses (encoding $b_0=1$),
$2\floor{n/2}$ binary clauses (encoding Proposition~\ref{prop:prod}), and
$n-1$ programmatic clauses (encoding Definition~\ref{def:cgp}).

Finally, we provide counts of the total number of complex Golay pairs of
length $n\leq\maxorder$ in Table~\ref{tbl:counts}.
The sizes of $\Omega_{\text{seqs}}$ and $\Omega_{\text{all}}$
match those from~\cite{gibson2011quaternary} in all cases
and the size of $\Omega_{\text{inequiv}}$ matches those from~\cite{CHK:DM:2002}
for $n\leq19$ (the largest length they exhaustively solved).

\begin{table}\begin{tabular}{cccc}
$n$ & $\abs{\Leven}$ & $\abs{\Lodd}$ & $\abs{\LA}$ \\ \hline
1          & 1          & $-$        & 1          \\ 
2          & 3          & 1          & 3          \\ 
3          & 3          & 1          & 1          \\ 
4          & 3          & 4          & 3          \\ 
5          & 12         & 4          & 5          \\ 
6          & 12         & 16         & 14         \\ 
7          & 39         & 16         & 12         \\ 
8          & 48         & 64         & 36         \\ 
9          & 153        & 64         & 44         \\ 
10         & 153        & 204        & 120        \\ 
11         & 561        & 252        & 101        \\ 
12         & 645        & 860        & 465        \\ 
13         & 2121       & 884        & 293        \\ 
14         & 2463       & 3284       & 317        \\ 
15         & 8340       & 3572       & 1793       \\ 
16         & 9087       & 12116      & 923        \\ 
17         & 31275      & 12824      & 3710       \\ 
18         & 34560      & 46080      & 14353      \\ 
19         & 117597     & 50944      & 10918      \\ 
20         & 130215     & 173620     & 26869      \\ 
21         & 446052     & 194004     & 116612     \\ 
22         & 500478     & 667304     & 67349      \\ 
23         & 1694865    & 732232     & 182989     \\ 
24         & 1886568    & 2515424    & 313878     \\ 
25         & 6447090    & 2727452    & 1211520
\end{tabular}
\caption{The number of sequences $\boldsymbol\Aeven$, $\boldsymbol\Aodd$,
and $\boldsymbol A$ that passed the filtering conditions of our algorithm
in lengths up to~$\boldsymbol\maxorder$.}
\label{tbl:listsizes}\end{table}

\begin{table}\begin{tabular}{cccc}
$n$ & $\abs{\Omega_{\text{seqs}}}$ & $\abs{\Omega_{\text{all}}}$ & $\abs{\Omega_{\text{inequiv}}}$ \\ \hline
1          & 4          & 16         & 1          \\ 
2          & 16         & 64         & 1          \\ 
3          & 16         & 128        & 1          \\ 
4          & 64         & 512        & 2          \\ 
5          & 64         & 512        & 1          \\ 
6          & 256        & 2048       & 3          \\ 
7          & 0          & 0          & 0          \\ 
8          & 768        & 6656       & 17         \\ 
9          & 0          & 0          & 0          \\ 
10         & 1536       & 12288      & 20         \\ 
11         & 64         & 512        & 1          \\ 
12         & 4608       & 36864      & 52         \\ 
13         & 64         & 512        & 1          \\ 
14         & 0          & 0          & 0          \\ 
15         & 0          & 0          & 0          \\ 
16         & 13312      & 106496     & 204        \\ 
17         & 0          & 0          & 0          \\ 
18         & 3072       & 24576      & 24         \\ 
19         & 0          & 0          & 0          \\ 
20         & 26880      & 215040     & 340        \\ 
21         & 0          & 0          & 0          \\ 
22         & 1024       & 8192       & 12         \\ 
23         & 0          & 0          & 0          \\ 
24         & 98304      & 786432     & 1056       \\ 
25         & 0          & 0          & 0          
\end{tabular}
\caption{The number complex Golay pairs in lengths up to~$\boldsymbol\maxorder$.
The table counts the number of individual sequences, the
number of pairs, and the number of pairs up to equivalence.}
\label{tbl:counts}\end{table}

Because~\cite{fiedler2013small,gibson2011quaternary,CHK:DM:2002} do not provide
implementations or timings for the enumerations they completed it is not possible
for us to compare the efficiency of our algorithm to previous algorithms.
However, we note that the results in this paper did not require
an exorbitant amount of computing resources.  If one has access to~25 modern CPU
cores then one can exhaustively enumerate all complex Golay pairs up to length~$25$
using our software in about a day and we estimate that increasing this
to length~$26$ would take another week. 
We note that Fiedler's paper~\cite{fiedler2013small} enumerates complex Golay pairs
to length~$28$.  It is not clear whether this was accomplished using more computing resources or
a more efficient algorithm, though we note that the preprocessing and stage~1
of our method is similar to Fiedler's method
with some differences in the filtering theorems.

\section{Future Work}

Besides increasing the length to which complex Golay pairs have been enumerated
there are a number of avenues for improvements which could be made in
future work.  As one example, we remark that
we have not exploited the algebraic structure
of complex Golay pairs revealed by Craigen, Holzmann, and Kharaghani~\cite{CHK:DM:2002}.
In particular, those authors prove a theorem which implies that if $p\equiv3\pmod{4}$
is a prime which divides $n$ and $A$ is a member of
a complex Golay pair of length $n$ then the polynomial $h_A$ is not irreducible
over $\FF_p(i)$.  Ensuring that this property holds
could be added to the filtering conditions which were used in stage~1.
In fact, the authors relate the factorization of $h_A$ over $\FF_p(i)$
to the factorization of $h_B$ over $\FF_p(i)$ for any complex Golay pair
$(A,B)$.  This factorization could potentially be used to perform stage~2
more efficiently, possibly supplementing or replacing the SAT solver entirely,
though it is unclear if such a method would perform better than our method in practice.
In any case, it would not be possible to apply their theorem in all lengths
(for example when $n$ is a power of $2$).

A second possible improvement could be to symbolically determine the value of
$z$ with $\abs{z}=1$ which maximizes $\abs{h_{A'}(z)}^2$
in the preprocessing step.  Once this value of $z$ is known
then $A'$ can be filtered if $\abs{h_{A'}(z)}^2>2n$ and if not then no
other value of $z$ needs to be tried.  This would save evaluating
$h_{A'}(z)$ at the points $z=e^{2\pi ij/N}$ for $j=0$, $\dotsc$, $N-1$
and would also increase the number of sequences which get filtered.
However, it is unclear if this method would be beneficial in practice due to
the overhead of maximizing $\abs{h_{A'}(z)}^2$ subject to $\abs{z}=1$.

Another possible improvement could be obtained by deriving further properties
like Proposition~\ref{prop:prod} that complex Golay pairs must satisfy.
We have performed some preliminary searches for such properties;
for example, consider the following property
which could be viewed as a strengthening of Proposition~\ref{prop:prod}:
\[ a_k\overline{a_{n-k-1}} = (-1)^{n+1}b_k\overline{b_{n-k-1}}
\qquad\text{for}\qquad\text{$k=1$, $\dotsc$, $n-2$}. \]
An examination of all complex Golay pairs up to length~$\maxorder$ reveals that
they all satisfy this property except for a \emph{single} complex Golay pair
up to equivalence. 
The only pair which doesn't satisfy this property is equivalent to 
\[ \paren[\big]{[1, 1, 1, -1, 1, 1, -1, 1], [1, i, i, -1, 1, -i, -i, -1]} \]
and was already singled out in~\cite{fiedler2008multi} for being special
as the only known example of what they call a
``cross-over'' Golay sequence pair.
Since a counterexample exists to this property there is no hope of proving
it in general, but perhaps a suitable generalization could be proven.

\section*{Acknowledgements}

This work was made possible by the facilities of the Shared Hierarchical 
Academic Research Computing Network (SHARCNET) and Compute/Calcul Canada.
The authors would also like to thank the anonymous reviewers whose comments
improved this article's clarity.

\balance
\bibliographystyle{ACM-Reference-Format}
\bibliography{issac-cgolay}


\begin{thebibliography}{31}


\ifx \showCODEN    \undefined \def \showCODEN     #1{\unskip}     \fi
\ifx \showDOI      \undefined \def \showDOI       #1{#1}\fi
\ifx \showISBNx    \undefined \def \showISBNx     #1{\unskip}     \fi
\ifx \showISBNxiii \undefined \def \showISBNxiii  #1{\unskip}     \fi
\ifx \showISSN     \undefined \def \showISSN      #1{\unskip}     \fi
\ifx \showLCCN     \undefined \def \showLCCN      #1{\unskip}     \fi
\ifx \shownote     \undefined \def \shownote      #1{#1}          \fi
\ifx \showarticletitle \undefined \def \showarticletitle #1{#1}   \fi
\ifx \showURL      \undefined \def \showURL       {\relax}        \fi
\providecommand\bibfield[2]{#2}
\providecommand\bibinfo[2]{#2}
\providecommand\natexlab[1]{#1}
\providecommand\showeprint[2][]{arXiv:#2}

\bibitem[\protect\citeauthoryear{\'Abrah\'am}{\'Abrah\'am}{2015}]%
        {abraham2015building}
\bibfield{author}{\bibinfo{person}{Erika \'Abrah\'am}.}
  \bibinfo{year}{2015}\natexlab{}.
\newblock \showarticletitle{{Building bridges between symbolic computation and
  satisfiability checking}}. In \bibinfo{booktitle}{\emph{Proceedings of the
  2015 ACM on International Symposium on Symbolic and Algebraic Computation}}.
  ACM, \bibinfo{address}{New York}, \bibinfo{pages}{1--6}.
\newblock


\bibitem[\protect\citeauthoryear{{\'A}brah{\'a}m, Abbott, Becker, Bigatti,
  Brain, Buchberger, Cimatti, Davenport, England, Fontaine, Forrest, Griggio,
  Kroening, Seiler, and Sturm}{{\'A}brah{\'a}m et~al\mbox{.}}{2016}]%
        {sc2}
\bibfield{author}{\bibinfo{person}{Erika {\'A}brah{\'a}m},
  \bibinfo{person}{John Abbott}, \bibinfo{person}{Bernd Becker},
  \bibinfo{person}{Anna~M. Bigatti}, \bibinfo{person}{Martin Brain},
  \bibinfo{person}{Bruno Buchberger}, \bibinfo{person}{Alessandro Cimatti},
  \bibinfo{person}{James~H. Davenport}, \bibinfo{person}{Matthew England},
  \bibinfo{person}{Pascal Fontaine}, \bibinfo{person}{Stephen Forrest},
  \bibinfo{person}{Alberto Griggio}, \bibinfo{person}{Daniel Kroening},
  \bibinfo{person}{Werner~M. Seiler}, {and} \bibinfo{person}{Thomas Sturm}.}
  \bibinfo{year}{2016}\natexlab{}.
\newblock \showarticletitle{{$\text{SC}^2$: Satisfiability Checking meets
  Symbolic Computation (Project Paper)}}. In
  \bibinfo{booktitle}{\emph{Intelligent Computer Mathematics: 9th International
  Conference, CICM 2016, Bialystok, Poland, July 25--29, 2016, Proceedings}}.
  \bibinfo{publisher}{Springer International Publishing},
  \bibinfo{address}{Cham}, \bibinfo{pages}{28--43}.
\newblock
\showISBNx{978-3-319-42547-4}
\newblock
\shownote{\url{http://www.sc-square.org/}.}


\bibitem[\protect\citeauthoryear{Bright}{Bright}{2017}]%
        {brightthesis}
\bibfield{author}{\bibinfo{person}{Curtis Bright}.}
  \bibinfo{year}{2017}\natexlab{}.
\newblock \emph{\bibinfo{title}{Computational Methods for Combinatorial and
  Number Theoretic Problems}}.
\newblock \bibinfo{thesistype}{Ph.D. Dissertation}. \bibinfo{school}{University
  of Waterloo}.
\newblock


\bibitem[\protect\citeauthoryear{Bright, Ganesh, Heinle, Kotsireas, Nejati, and
  Czarnecki}{Bright et~al\mbox{.}}{2016}]%
        {bright2016mathcheck}
\bibfield{author}{\bibinfo{person}{Curtis Bright}, \bibinfo{person}{Vijay
  Ganesh}, \bibinfo{person}{Albert Heinle}, \bibinfo{person}{Ilias~S.
  Kotsireas}, \bibinfo{person}{Saeed Nejati}, {and} \bibinfo{person}{Krzysztof
  Czarnecki}.} \bibinfo{year}{2016}\natexlab{}.
\newblock \showarticletitle{{\textsc{MathCheck2}: {A} {SAT+CAS} Verifier for
  Combinatorial Conjectures}}. In \bibinfo{booktitle}{\emph{Computer Algebra in
  Scientific Computing - 18th International Workshop, {CASC} 2016, Bucharest,
  Romania, September 19--23, 2016, Proceedings}}. \bibinfo{pages}{117--133}.
\newblock


\bibitem[\protect\citeauthoryear{Bright, Kotsireas, and Ganesh}{Bright
  et~al\mbox{.}}{2018a}]%
        {bright2017sat+}
\bibfield{author}{\bibinfo{person}{Curtis Bright}, \bibinfo{person}{Ilias
  Kotsireas}, {and} \bibinfo{person}{Vijay Ganesh}.}
  \bibinfo{year}{2018}\natexlab{a}.
\newblock \showarticletitle{A {SAT+CAS} Method for Enumerating {W}illiamson
  Matrices of Even Order}. In \bibinfo{booktitle}{\emph{Proceedings of the
  Thirty-Second {AAAI} Conference on Artificial Intelligence}}.
\newblock


\bibitem[\protect\citeauthoryear{Bright, Kotsireas, Heinle, and Ganesh}{Bright
  et~al\mbox{.}}{2018b}]%
        {gcpweb}
\bibfield{author}{\bibinfo{person}{Curtis Bright}, \bibinfo{person}{Ilias
  Kotsireas}, \bibinfo{person}{Albert Heinle}, {and} \bibinfo{person}{Vijay
  Ganesh}.} \bibinfo{year}{2018}\natexlab{b}.
\newblock \bibinfo{title}{Complex Golay Pairs via SAT}.
\newblock
  \bibinfo{howpublished}{\url{https://cs.uwaterloo.ca/~cbright/cgpsat/}}.
\newblock
\newblock
\shownote{Complex {G}olay pairs archived at
  \url{https://zenodo.org/record/1246337}, code available at
  \url{https://bitbucket.org/cbright/mathcheck2}.}


\bibitem[\protect\citeauthoryear{Craigen}{Craigen}{1994}]%
        {C:JCMCC:1994}
\bibfield{author}{\bibinfo{person}{R. Craigen}.}
  \bibinfo{year}{1994}\natexlab{}.
\newblock \showarticletitle{Complex {G}olay sequences}.
\newblock \bibinfo{journal}{\emph{J. Combin. Math. Combin. Comput.}}
  \bibinfo{volume}{15} (\bibinfo{year}{1994}), \bibinfo{pages}{161--169}.
\newblock
\showISSN{0835-3026}


\bibitem[\protect\citeauthoryear{Craigen, Holzmann, and Kharaghani}{Craigen
  et~al\mbox{.}}{2002}]%
        {CHK:DM:2002}
\bibfield{author}{\bibinfo{person}{R. Craigen}, \bibinfo{person}{W. Holzmann},
  {and} \bibinfo{person}{H. Kharaghani}.} \bibinfo{year}{2002}\natexlab{}.
\newblock \showarticletitle{Complex {G}olay sequences: structure and
  applications}.
\newblock \bibinfo{journal}{\emph{Discrete Math.}} \bibinfo{volume}{252},
  \bibinfo{number}{1-3} (\bibinfo{year}{2002}), \bibinfo{pages}{73--89}.
\newblock
\showCODEN{DSMHA4}
\showISSN{0012-365X}


\bibitem[\protect\citeauthoryear{Davis and Jedwab}{Davis and Jedwab}{1999}]%
        {davis1999peak}
\bibfield{author}{\bibinfo{person}{James~A Davis} {and}
  \bibinfo{person}{Jonathan Jedwab}.} \bibinfo{year}{1999}\natexlab{}.
\newblock \showarticletitle{Peak-to-mean power control in {OFDM}, {G}olay
  complementary sequences, and {R}eed--{M}uller codes}.
\newblock \bibinfo{journal}{\emph{IEEE Transactions on Information Theory}}
  \bibinfo{volume}{45}, \bibinfo{number}{7} (\bibinfo{year}{1999}),
  \bibinfo{pages}{2397--2417}.
\newblock


\bibitem[\protect\citeauthoryear{Donato, Urena, Mazo, and Alvarez}{Donato
  et~al\mbox{.}}{2004}]%
        {1336500}
\bibfield{author}{\bibinfo{person}{P.~G. Donato}, \bibinfo{person}{J. Urena},
  \bibinfo{person}{M. Mazo}, {and} \bibinfo{person}{F. Alvarez}.}
  \bibinfo{year}{2004}\natexlab{}.
\newblock \showarticletitle{Train wheel detection without electronic equipment
  near the rail line}. In \bibinfo{booktitle}{\emph{IEEE Intelligent Vehicles
  Symposium, 2004}}. \bibinfo{pages}{876--880}.
\newblock
\urldef\tempurl%
\url{https://doi.org/10.1109/IVS.2004.1336500}
\showDOI{\tempurl}


\bibitem[\protect\citeauthoryear{Fiedler}{Fiedler}{2013}]%
        {fiedler2013small}
\bibfield{author}{\bibinfo{person}{Frank Fiedler}.}
  \bibinfo{year}{2013}\natexlab{}.
\newblock \showarticletitle{Small {G}olay sequences.}
\newblock \bibinfo{journal}{\emph{Advances in Mathematics of Communications}}
  \bibinfo{volume}{7}, \bibinfo{number}{4} (\bibinfo{year}{2013}).
\newblock


\bibitem[\protect\citeauthoryear{Fiedler, Jedwab, and Parker}{Fiedler
  et~al\mbox{.}}{2008a}]%
        {fiedler2008framework}
\bibfield{author}{\bibinfo{person}{Frank Fiedler}, \bibinfo{person}{Jonathan
  Jedwab}, {and} \bibinfo{person}{Matthew~G Parker}.}
  \bibinfo{year}{2008}\natexlab{a}.
\newblock \showarticletitle{A Framework for the Construction of {G}olay
  Sequences}.
\newblock \bibinfo{journal}{\emph{IEEE Transactions on Information Theory}}
  \bibinfo{volume}{54}, \bibinfo{number}{7} (\bibinfo{year}{2008}),
  \bibinfo{pages}{3114--3129}.
\newblock


\bibitem[\protect\citeauthoryear{Fiedler, Jedwab, and Parker}{Fiedler
  et~al\mbox{.}}{2008b}]%
        {fiedler2008multi}
\bibfield{author}{\bibinfo{person}{Frank Fiedler}, \bibinfo{person}{Jonathan
  Jedwab}, {and} \bibinfo{person}{Matthew~G Parker}.}
  \bibinfo{year}{2008}\natexlab{b}.
\newblock \showarticletitle{A multi-dimensional approach to the construction
  and enumeration of Golay complementary sequences}.
\newblock \bibinfo{journal}{\emph{Journal of Combinatorial Theory, Series A}}
  \bibinfo{volume}{115}, \bibinfo{number}{5} (\bibinfo{year}{2008}),
  \bibinfo{pages}{753--776}.
\newblock


\bibitem[\protect\citeauthoryear{Frigo and Johnson}{Frigo and Johnson}{2005}]%
        {frigo2005design}
\bibfield{author}{\bibinfo{person}{Matteo Frigo} {and}
  \bibinfo{person}{Steven~G Johnson}.} \bibinfo{year}{2005}\natexlab{}.
\newblock \showarticletitle{The design and implementation of FFTW3}.
\newblock \bibinfo{journal}{\emph{Proc. IEEE}} \bibinfo{volume}{93},
  \bibinfo{number}{2} (\bibinfo{year}{2005}), \bibinfo{pages}{216--231}.
\newblock


\bibitem[\protect\citeauthoryear{Ganesh, O'Donnell, Soos, Devadas, Rinard, and
  Solar-Lezama}{Ganesh et~al\mbox{.}}{2012}]%
        {ganesh2012lynx}
\bibfield{author}{\bibinfo{person}{Vijay Ganesh}, \bibinfo{person}{Charles~W
  O'Donnell}, \bibinfo{person}{Mate Soos}, \bibinfo{person}{Srinivas Devadas},
  \bibinfo{person}{Martin~C Rinard}, {and} \bibinfo{person}{Armando
  Solar-Lezama}.} \bibinfo{year}{2012}\natexlab{}.
\newblock \showarticletitle{Lynx: A programmatic {SAT} solver for the
  {RNA}-folding problem}. In \bibinfo{booktitle}{\emph{International Conference
  on Theory and Applications of Satisfiability Testing}}. Springer,
  \bibinfo{pages}{143--156}.
\newblock


\bibitem[\protect\citeauthoryear{Gibson and Jedwab}{Gibson and Jedwab}{2011}]%
        {gibson2011quaternary}
\bibfield{author}{\bibinfo{person}{Richard~G Gibson} {and}
  \bibinfo{person}{Jonathan Jedwab}.} \bibinfo{year}{2011}\natexlab{}.
\newblock \showarticletitle{Quaternary {G}olay sequence pairs {I}: {E}ven
  length}.
\newblock \bibinfo{journal}{\emph{Designs, Codes and Cryptography}}
  \bibinfo{volume}{59}, \bibinfo{number}{1-3} (\bibinfo{year}{2011}),
  \bibinfo{pages}{131--146}.
\newblock


\bibitem[\protect\citeauthoryear{Golay}{Golay}{1961}]%
        {golay1961complementary}
\bibfield{author}{\bibinfo{person}{Marcel Golay}.}
  \bibinfo{year}{1961}\natexlab{}.
\newblock \showarticletitle{Complementary series}.
\newblock \bibinfo{journal}{\emph{IRE Transactions on Information Theory}}
  \bibinfo{volume}{7}, \bibinfo{number}{2} (\bibinfo{year}{1961}),
  \bibinfo{pages}{82--87}.
\newblock


\bibitem[\protect\citeauthoryear{Golay}{Golay}{1949}]%
        {golay1949multi}
\bibfield{author}{\bibinfo{person}{Marcel~J.E. Golay}.}
  \bibinfo{year}{1949}\natexlab{}.
\newblock \showarticletitle{Multi-slit spectrometry}.
\newblock \bibinfo{journal}{\emph{JOSA}} \bibinfo{volume}{39},
  \bibinfo{number}{6} (\bibinfo{year}{1949}), \bibinfo{pages}{437--444}.
\newblock


\bibitem[\protect\citeauthoryear{Holzmann and Kharaghani}{Holzmann and
  Kharaghani}{1994}]%
        {HK:AJC:1994}
\bibfield{author}{\bibinfo{person}{W.~H. Holzmann} {and} \bibinfo{person}{H.
  Kharaghani}.} \bibinfo{year}{1994}\natexlab{}.
\newblock \showarticletitle{A computer search for complex {G}olay sequences}.
\newblock \bibinfo{journal}{\emph{Australas. J. Combin.}}  \bibinfo{volume}{10}
  (\bibinfo{year}{1994}), \bibinfo{pages}{251--258}.
\newblock
\showISSN{1034-4942}


\bibitem[\protect\citeauthoryear{Hussain, Khan, Khalid, and Iqbal}{Hussain
  et~al\mbox{.}}{2014}]%
        {Hussain2014}
\bibfield{author}{\bibinfo{person}{Aamir Hussain}, \bibinfo{person}{Zeashan~H.
  Khan}, \bibinfo{person}{Azfar Khalid}, {and} \bibinfo{person}{Muhammad
  Iqbal}.} \bibinfo{year}{2014}\natexlab{}.
\newblock \bibinfo{booktitle}{\emph{A Comparison of Pulse Compression
  Techniques for Ranging Applications}}.
\newblock \bibinfo{publisher}{Springer Singapore},
  \bibinfo{address}{Singapore}, \bibinfo{pages}{169--191}.
\newblock
\showISBNx{978-981-4585-36-1}
\urldef\tempurl%
\url{https://doi.org/10.1007/978-981-4585-36-1_5}
\showDOI{\tempurl}


\bibitem[\protect\citeauthoryear{Kharaghani and Tayfeh-Rezaie}{Kharaghani and
  Tayfeh-Rezaie}{2005}]%
        {kharaghani2005hadamard}
\bibfield{author}{\bibinfo{person}{Hadi Kharaghani} {and}
  \bibinfo{person}{Behruz Tayfeh-Rezaie}.} \bibinfo{year}{2005}\natexlab{}.
\newblock \showarticletitle{A {H}adamard matrix of order 428}.
\newblock \bibinfo{journal}{\emph{Journal of Combinatorial Designs}}
  \bibinfo{volume}{13}, \bibinfo{number}{6} (\bibinfo{year}{2005}),
  \bibinfo{pages}{435--440}.
\newblock


\bibitem[\protect\citeauthoryear{Kotsireas}{Kotsireas}{2013}]%
        {kotsireas2013algorithms}
\bibfield{author}{\bibinfo{person}{Ilias~S Kotsireas}.}
  \bibinfo{year}{2013}\natexlab{}.
\newblock \showarticletitle{Algorithms and metaheuristics for combinatorial
  matrices}.
\newblock In \bibinfo{booktitle}{\emph{Handbook of Combinatorial
  Optimization}}. \bibinfo{publisher}{Springer}, \bibinfo{pages}{283--309}.
\newblock


\bibitem[\protect\citeauthoryear{Li and Chu}{Li and Chu}{2005}]%
        {li2005more}
\bibfield{author}{\bibinfo{person}{Ying Li} {and} \bibinfo{person}{Wen~Bin
  Chu}.} \bibinfo{year}{2005}\natexlab{}.
\newblock \showarticletitle{More Golay sequences}.
\newblock \bibinfo{journal}{\emph{IEEE Transactions on Information Theory}}
  \bibinfo{volume}{51}, \bibinfo{number}{3} (\bibinfo{year}{2005}),
  \bibinfo{pages}{1141--1145}.
\newblock


\bibitem[\protect\citeauthoryear{Liang, Poupart, Czarnecki, and Ganesh}{Liang
  et~al\mbox{.}}{2017}]%
        {liang2017empirical}
\bibfield{author}{\bibinfo{person}{Jia~Hui Liang}, \bibinfo{person}{Pascal
  Poupart}, \bibinfo{person}{Krzysztof Czarnecki}, {and} \bibinfo{person}{Vijay
  Ganesh}.} \bibinfo{year}{2017}\natexlab{}.
\newblock \showarticletitle{An Empirical Study of Branching Heuristics Through
  the Lens of Global Learning Rate}. In \bibinfo{booktitle}{\emph{International
  Conference on Theory and Applications of Satisfiability Testing}}. Springer,
  \bibinfo{pages}{119--135}.
\newblock


\bibitem[\protect\citeauthoryear{Lomayev, Gagiev, Maltsev, Kasher, Genossar,
  and Cordeiro}{Lomayev et~al\mbox{.}}{2017}]%
        {lomayev2017golay}
\bibfield{author}{\bibinfo{person}{A. Lomayev}, \bibinfo{person}{Y.P. Gagiev},
  \bibinfo{person}{A. Maltsev}, \bibinfo{person}{A. Kasher},
  \bibinfo{person}{M. Genossar}, {and} \bibinfo{person}{C. Cordeiro}.}
  \bibinfo{year}{2017}\natexlab{}.
\newblock \bibinfo{title}{Golay sequences for wireless networks}.
\newblock
\newblock
\urldef\tempurl%
\url{https://www.google.com/patents/US20170324461}
\showURL{%
\tempurl}
\newblock
\shownote{US Patent App. 15/280,635.}


\bibitem[\protect\citeauthoryear{Monagan, Geddes, Heal, Labahn, Vorkoetter,
  McCarron, and DeMarco}{Monagan et~al\mbox{.}}{2005}]%
        {Maple10}
\bibfield{author}{\bibinfo{person}{Michael~B. Monagan},
  \bibinfo{person}{Keith~O. Geddes}, \bibinfo{person}{K.~Michael Heal},
  \bibinfo{person}{George Labahn}, \bibinfo{person}{Stefan~M. Vorkoetter},
  \bibinfo{person}{James McCarron}, {and} \bibinfo{person}{Paul DeMarco}.}
  \bibinfo{year}{2005}\natexlab{}.
\newblock \bibinfo{booktitle}{\emph{Maple~10 Programming Guide}}.
\newblock \bibinfo{publisher}{Maplesoft}, \bibinfo{address}{Waterloo ON,
  Canada}.
\newblock


\bibitem[\protect\citeauthoryear{Nazarathy, Newton, Giffard, Moberly, Sischka,
  Trutna, and Foster}{Nazarathy et~al\mbox{.}}{1989}]%
        {nazarathy1989real}
\bibfield{author}{\bibinfo{person}{Moshe Nazarathy}, \bibinfo{person}{Steven~A
  Newton}, \bibinfo{person}{RP Giffard}, \bibinfo{person}{DS Moberly},
  \bibinfo{person}{F Sischka}, \bibinfo{person}{WR Trutna}, {and}
  \bibinfo{person}{S Foster}.} \bibinfo{year}{1989}\natexlab{}.
\newblock \showarticletitle{Real-time long range complementary correlation
  optical time domain reflectometer}.
\newblock \bibinfo{journal}{\emph{Journal of Lightwave Technology}}
  \bibinfo{volume}{7}, \bibinfo{number}{1} (\bibinfo{year}{1989}),
  \bibinfo{pages}{24--38}.
\newblock


\bibitem[\protect\citeauthoryear{Nowicki, Secomski, Litniewski, Trots, and
  Lewin}{Nowicki et~al\mbox{.}}{2003}]%
        {nowicki2003application}
\bibfield{author}{\bibinfo{person}{A Nowicki}, \bibinfo{person}{W Secomski},
  \bibinfo{person}{J Litniewski}, \bibinfo{person}{I Trots}, {and}
  \bibinfo{person}{PA Lewin}.} \bibinfo{year}{2003}\natexlab{}.
\newblock \showarticletitle{On the application of signal compression using
  {G}olay's codes sequences in ultrasound diagnostic}.
\newblock \bibinfo{journal}{\emph{Archives of Acoustics}} \bibinfo{volume}{28},
  \bibinfo{number}{4} (\bibinfo{year}{2003}).
\newblock


\bibitem[\protect\citeauthoryear{Paterson}{Paterson}{2000}]%
        {paterson2000generalized}
\bibfield{author}{\bibinfo{person}{Kenneth~G Paterson}.}
  \bibinfo{year}{2000}\natexlab{}.
\newblock \showarticletitle{Generalized {R}eed--{M}uller codes and power
  control in {OFDM} modulation}.
\newblock \bibinfo{journal}{\emph{IEEE Transactions on Information Theory}}
  \bibinfo{volume}{46}, \bibinfo{number}{1} (\bibinfo{year}{2000}),
  \bibinfo{pages}{104--120}.
\newblock


\bibitem[\protect\citeauthoryear{Riel}{Riel}{2006}]%
        {nsoks}
\bibfield{author}{\bibinfo{person}{Joe Riel}.} \bibinfo{year}{2006}\natexlab{}.
\newblock \bibinfo{title}{\texttt{nsoks}: A \textsc{Maple} script for writing
  $n$ as a sum of $k$ squares}.
\newblock \bibinfo{howpublished}{\url{http://www.swmath.org/software/21060}}.
\newblock


\bibitem[\protect\citeauthoryear{Zulkoski, Bright, Heinle, Kotsireas,
  Czarnecki, and Ganesh}{Zulkoski et~al\mbox{.}}{2017}]%
        {DBLP:journals/jar/ZulkoskiBHKCG17}
\bibfield{author}{\bibinfo{person}{Edward Zulkoski}, \bibinfo{person}{Curtis
  Bright}, \bibinfo{person}{Albert Heinle}, \bibinfo{person}{Ilias~S.
  Kotsireas}, \bibinfo{person}{Krzysztof Czarnecki}, {and}
  \bibinfo{person}{Vijay Ganesh}.} \bibinfo{year}{2017}\natexlab{}.
\newblock \showarticletitle{Combining {SAT} Solvers with Computer Algebra
  Systems to Verify Combinatorial Conjectures}.
\newblock \bibinfo{journal}{\emph{J. Autom. Reasoning}} \bibinfo{volume}{58},
  \bibinfo{number}{3} (\bibinfo{year}{2017}), \bibinfo{pages}{313--339}.
\newblock
\urldef\tempurl%
\url{https://doi.org/10.1007/s10817-016-9396-y}
\showDOI{\tempurl}


\end{thebibliography}

\end{document}